\documentclass[12pt,english]{iopart} 

\expandafter\let\csname equation*\endcsname\relax
\expandafter\let\csname endequation*\endcsname\relax

\usepackage{graphicx} 
\usepackage{dcolumn}   
\usepackage{bm}        
\usepackage{multirow}
\usepackage[export]{adjustbox}
\usepackage{amssymb}
\usepackage{amsfonts}   
\usepackage{amsmath} 
\usepackage{amsthm}
\usepackage{cite}
\usepackage{color}
\usepackage{marvosym}
\newcommand{\mathbbm}[1]{\text{\usefont{U}{bbm}{m}{n}#1}} 
\newtheorem{theorem}{Theorem}
\newtheorem{corollary}[theorem]{Corollary}

\newcommand{\bra}[1]{\langle #1|}
\newcommand{\ket}[1]{|#1\rangle}

\newcommand{\ketbra}[1]{| #1\rangle \langle #1|}

\newcommand{\be}{\begin{equation}}
\newcommand{\ee}{\end{equation}}
\newcommand{\bea}{\begin{eqnarray}}
\newcommand{\eea}{\end{eqnarray}}

\newcommand{\eins}{\mathbbm{1}}

\newcommand{\kommentar}[1]{}

\newtheorem{lemma}[theorem]{Lemma}

\newcommand{\forget}[1]{}

\begin{document}

\title{Graphical description of unitary transformations on hypergraph states}

\author{
Mariami Gachechiladze$^1$,
Nikoloz Tsimakuridze$^2$,\\
and Otfried G\"uhne$^1$
}

\address{$^1$Naturwissenschaftlich-Technische Fakult\"at,
Universit\"at Siegen,
Walter-Flex-Stra{\ss}e~3,
57068 Siegen, Germany}

\address{$^2$School of Mathematics and Computer Science, 
Free University of Tbilisi, \\
240 David Agmashenebeli alley, 0159 Tbilisi,
Georgia
}

\date{\today}

\begin{abstract}
Hypergraph states form a family of multiparticle quantum states 
that generalizes cluster states and graph states. We study the action 
and graphical representation of nonlocal unitary transformations between 
hypergraph states. This leads to a generalization of local complementation 
and graphical rules for various gates, such as the CNOT gate and the Toffoli 
gate. As an application, we show that already for five qubits local Pauli 
operations are not sufficient to check local equivalence of hypergraph 
states. Furthermore, we use our rules to construct
entanglement witnesses for three-uniform hypergraph states.
\end{abstract}


\section{Introduction}
Due to its possible applications in quantum information processing,
multiparticle entanglement is under intensive research. One of the
problems in this field is the identification of families of states which
are useful in applications, but nevertheless can be described by a simple
formalism.  An interesting class of multi-qubit quantum states are graph 
states \cite{hein}. Mathematically, these states are described by graphs, 
where the vertices correspond to particles and the edges represent two-body
interactions in a possible generation process. A generalization of these states
are hypergraph states \cite{Kruszynska2009, Qu2013_encoding, rossinjp, Otfried}.
In a hypergraph, an edge can connect more than two vertices, so hypergraph states
can be generated with multi-qubit interactions. Hypergraph states have 
turned out to violate local realism in a robust manner \cite{gbg}, they 
play a role in quantum algorithms \cite{scripta} and are central for novel 
schemes of measurement-based quantum computation \cite{akimasa}.

A general feature of graph and hypergraph states is that different graphs may
lead to quantum states with the same entanglement properties. It is therefore
important to study the action of local and nonlocal unitary transformations between
these states. For graph states, the so-called local complementation plays an outstanding
role \cite{nestgraphicaldescription}: This graphical transformation corresponds to 
so-called local Clifford operations and these operations represent all possible local unitary 
transformations between graph states for up to eight qubits \cite{adan8pra}, only
for large qubit numbers other transformations play a role \cite{Ji2010_LU-LC, nika}.

In this paper we derive graphical rules to represent various unitary transformations
between hypergraph states. First, we introduce a generalization of local complementation
to hypergraphs and the corresponding unitary transformations. Then we consider different
quantum gates, such as the CNOT and Toffoli gate and their graphical representation. In 
general, the considered unitary transformations are nonlocal, but in some cases they can be
combined to give effectively local transformations. With that, we find pairs of five-qubit hypergraph states, 
which are equivalent under local unitary transformations, but they are not equivalent under
local Pauli operations. These are the first examples of this kind, up to four qubits all 
locally equivalent hypergraph states could be transformed into each other by application
of Pauli matrices only \cite{Otfried, chenlei}. As a second application, we construct entanglement witnesses for 
hypergraph states which contain only three-edges. This will be useful for 
characterizing entanglement in these states experimentally. 

\section{Local complementation of hypergraph states}

\subsection{Basic definitions and  local complementation of graphs}

Let us start by defining hypergraph states, a detailed discussion of
their properties can be found in Ref.~\cite{Otfried}. A hypergraph 
$H=(V,E)$ consists of a set of vertices $V=\{1,\dots,N\}$ and a set 
of hyperedges $E\subset 2^V$, with $2^V$ being the power set of $V$, 
some examples of hypergraphs can be found in Fig.~\ref{fig1}. While 
for graphs edges connect exactly two vertices, hyperedges can connect 
more than two vertices, or contain just a single vertex. For any hypergraph the corresponding 
hypergraph state $\ket{H}$ is defined as the $N$-qubit state 
\begin{equation}
\ket{H} = \prod_{e\in E} C_e\ket{+}^{\otimes N},
\label{eq-hg-creation}
\end{equation}
where $\ket{+}=(\ket{0}+\ket{1})/\sqrt{2}$ are the initial single-qubit 
states, $e\in E$ is a hyperedge and $C_e$ is a multi-qubit phase gate 
$C_e =\mathbbm {1}  - 2\ket{1\dots 1}\bra{1\dots 1}$, acting on the Hilbert 
space associated with the vertices $v\in e$. Since all these phase gates 
commute, the order in the product does not matter. It is useful to note 
that hypergraph states are exactly the states that can be written as
\begin{equation}
\ket{H} = \sum_{x\in\{0,1\}^n}(-1)^{f(x)}\ket{x},
\label{eq-hg-equally}
\end{equation}
with $f(x)\in \{0,1\}$ being some binary function. {From} this representation, 
one recognizes that hypergraph states are special cases of locally maximally 
entanglable (LME) states. LME states are wider classes of states as they 
allow arbitrary equal complex phases (corresponding to an arbitrary $f(x)$)
in the full computational basis.  The  name LME is due to the fact that they 
are maximally entangleable to auxiliary systems using only local operations 
\cite{Kruszynska2009}.

A very important subclass  of hypergraph states are graph states.  
Their properties and applications have been studied extensively.  
Graph states correspond to graphs and therefore, only two-body 
controlled phase gates $C_{\{ij\}}$ are required for their generation. 
They are local stabiliser states and prominent examples of them  
are Greenberger-Horne-Zeilinger states and cluster states. For a review on graphs states  
we direct the reader to  Ref.~\cite{hein}.

Once it is established that graph states are important classes of 
multiqubit states, it is crucial to learn which of these states 
are equivalent under local actions of each party. As local actions 
one considers here local unitary transformations. Here the discrete
subclass of local Clifford operations play an outstanding role. By
definition, local Clifford operations leave the set of Pauli matrices invariant.
It has been shown that a graph state $\ket{G}$ can be transformed to 
another graph state $\ket{G'}$ by means of local Clifford action on 
some parties, if the graph $G'$ can be obtained from the graph $G$ 
by a series of local complementations \cite{nestgraphicaldescription}. 
The local complementation 
of a graph $G$ works as follows: One picks a vertex $a\in V$ and 
complements then subgraph in the neighbourhood $\mathcal{N}(a)$ of $a$, 
defined as the set containing all adjacent vertices to $a$. The 
complementation means that vertices in the neighbourhood become 
disconnected, if they were connected before, and they become connected, 
if they were disconnected before. Originally, the rule of local 
complementation has been conjectured to be necessary and sufficient 
for local unitary equivalence of graph states \cite{nestgraphicaldescription}, however, it 
was later disproved  by counterexamples \cite{Ji2010_LU-LC, nika}.

In order to physically achieve the local unitary transformation 
corresponding local complementation, the following unitary transform 
is considered:
\begin{equation}
\tau_g(a) = \sqrt{X_{a}}^\pm\underset{_{b\in N(a)}}{\prod}\sqrt{Z_{b}}^\mp,
\label{eq-lc-graph}
\end{equation}
where $X$ and $Z$ denote the Pauli matrices, and 
\begin{equation}
\sqrt{X}^\pm = \begin{pmatrix}
\frac{1 \pm i}{2} & \frac{1 \mp i}{2} \\
\frac{1 \mp i}{2} & \frac{1 \pm i}{2}
\end{pmatrix}=
\ketbra{+}\pm i\ketbra{-},
\qquad
\sqrt{Z}^\pm = \begin{pmatrix}
1 & 0 \\
0 & \pm i
\end{pmatrix}.
\end{equation}
In the next section, we will generalize this to hypergraph states.

\subsection{Local complementation of hypergraphs}
Now we extend the term local complementation and its action to all 
hypergraph states. We denote the set of vertices of a hypergraph to be 
$V$ and the set of edges to be $E$. First we introduce a term which can be regarded 
as a generalization of the term neighbourhood known in graph theory. 
We call it \textit{adjacency} of a vertex $a\in V$ and denote it 
by  $\mathcal{A}(a)=\{e-\{a\}|e\in E \mbox{ with } a\in e \}$. 
The elements of $\mathcal{A}(a)$ are sets of vertices which are
adjacent to $a$ via some hyperedge. To give an an example, the 
adjacency of the vertex  $a=1$ from the  hypergraph in the top image 
of Fig.~\ref{fig1}~(a) is given by $\mathcal{A}(1)=\{\{3\},\{2,3\},\{4,5\}\}$.  
Similarly we can define the adjacency for some set of vertices  
$W\subseteq V$ as  $\mathcal{A}(W)=\{e-W| e\in E \mbox{ with } 
{W} \subseteq e\}$.

For formulating our main result, we introduce the concept of a local 
edge-pair complementation in hypergraphs around a vertex $a\in V$. 
Let us define first the set of \textit{adjacency pairs} of vertex 
$a$ to be the set $\mathcal{A}_2(a) = \{ \{e_1, e_2\} | e_1 \neq e_2,
e_1 \in \mathcal{A}(a), e_2 \in \mathcal{A}(a) \}$ of all distinct pairs 
in the adjacency set. Considering again the top image from the 
Fig.~\ref{fig1}~(a),  
the set of adjacency pairs for vertex $a=1$ is given by the set  
$\mathcal{A}_2(1)=\{\{\{3\},\{2,3\}\},\{\{3\},\{4,5\}\}, \{\{2,3\},\{4,5\}\}\}$. 
Finally, the {\it local edge-pair complementation} around a vertex $a$ complements 
the edges in the multiset $P = \{ e_1 \cup e_2 | \{e_1, e_2\} \in \mathcal{A}_2(a)\}$.
Notice that $P$ is a multiset and only the edges appearing with odd multiplicity 
will be affected. We again consider the top image from Fig.~\ref{fig1}~(a) for which 
the multiset is $P = \{\{2,3\},\{3,4,5\},\{2,3,4,5\}\}$. Complementation of the edges
in this multiset means that they are deleted from the hypergraph, if they were already
present, and the are added, if they were not present. 

In the following theorem we show that a local edge-pair complementation transforming  
a hypergraph state $\ket{H}$ to a hypergraph state $\ket{H'}$ around a vertex $a$ 
can be achieved by the the following nonlocal operation:
\begin{equation}
\tau(a)= \sqrt{X_{a}}^\pm\underset{e \in \mathcal{A}(a)}{\prod}\sqrt{C_{e}}^\mp,
\label{eq-lc-hypergraph}
\end{equation}
where $\sqrt{C_{e}}^\pm=\mathbbm{1}-[1-(\pm i)]\ketbra{11\dots 1}$  is the square
root of controlled phase gate applied to qubits in edge $e$. It is a diagonal 
operator with every eigenvalue being one, except for the eigenvalue corresponding 
to state $\ket{1}^{\otimes |e|}$, which is $\pm i$. In the following, we sometimes 
write $\tau^\pm(a)$ in order to indicate the sign of $\sqrt{X_{a}}^\pm.$
Note that for a usual graph Eq.~(\ref{eq-lc-hypergraph}) corresponds to 
Eq.~(\ref{eq-lc-graph}) and local edge-pair complementation corresponds to the local complementation. We can 
formulate:

\begin{figure}[t]
\begin{center}
\includegraphics[width=0.85\textwidth, right]{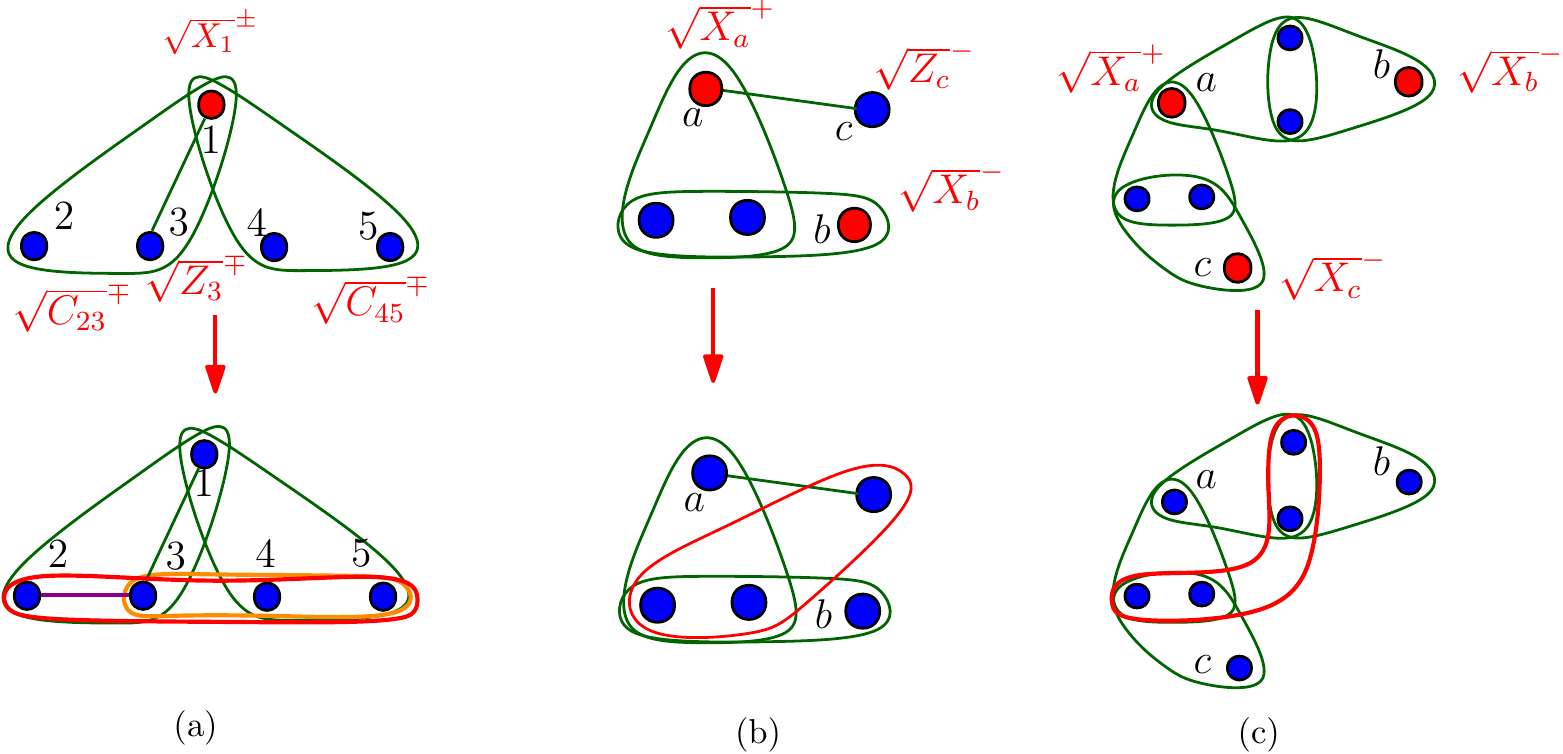}
\end{center}
\caption{The top row  shows hypergraphs before application of the extended 
local complementation rules. The bottom row shows resulting hypergraphs 
after the transformation has been made. (a) An example of a hypergraph 
state. Here set of vertices are $\{1,2,3,4,5\}$ , the set of edges are 
$E=\{\{1,2,3\},\{1,3\},\{1,4,5\}\}$, the set of adjacencies for qubit $1$ 
is $\mathcal{A}(1)=\{\{3\},\{2,3\},\{4,5\}\}$ and  the adjacency pairs are  
$\mathcal{A}_2(1)=\{\{\{3\},\{2,3\}\},\{\{3\},\{4,5\}\}, \{\{2,3\},\{4,5\}\}\}$ 
Therefore, the set of complemented edges are $P=\{\{2,3\},\{3,4,5\},\{2,3,4,5\}\}$. 
After application of $\tau(1)$ the edges from the multiset $P$ are complemented 
and in this case all three new edges are created. 
(b) Here, we apply two transformations, $\tau^+(a)$ and $\tau^-(b)$.
As a consequence, the square roots of the two-qubit phase gates cancel out
and we are left with local Clifford operations. This example demonstrates 
that local Pauli operations are not enough to exhaust all equivalence 
classes already in five-qubit hypergraph states. 
(c) An example of application of three transformations. Again, phase gates 
on the neighbourhoods of vertices $a$, $b$, and $c$  cancel out and a local
transformation remains.}
\label{fig1}
\end{figure} 
 
\begin{theorem}
For any hypergraph state the transformation $\tau(a)$ around a 
vertex $a \in V$  performs a local edge-pair complementation on 
its corresponding hypergraph.
\end{theorem}

\begin{proof}
Without loss of generality we assume that $a=N$ is the last qubit.
Then, we need to fix the following notation. First, let $x$ or $\ket{x}$ 
be an element of the computational basis on the first $N-1$ qubits, $x$ 
can be seen as a string of $0$ and $1$. We say that an edge $e$ acts on 
$x$ if $x$ has  the entries $1$ on the qubits belonging to $e$. 
This means that the phase gate $C_e$ changes the sign of $\ket{x}.$ We 
define $m_x=| \{e \in H | e \mbox{ acts on }x\}|$ as the number edges 
acting on $x$. Since $x$ is defined on $N-1$ qubits only, we also define 
$m^N_x = |\{e \in H | e \mbox{ acts on } \{x, 1_N\} \}|$
as the number of edges that act on $\{x, 1_N\}$. This means that the 
edges act on the basis element, where a 1 is appended
as the state of the last qubit. 

Using this notation, the hypergraph state can be rewritten as: 
$\ket{H}=\sum_x (-1)^{m_x}\big[\ket{x}\otimes(\ket{0}+(-1)^{m_x^N}\ket{1})\big]$
and we can compute:
\begin{align}
\tau^{+}(N) \ket{H} =& \sqrt{X_N}^\pm \underset{e \in \mathcal{A}(N)}{\prod}\sqrt{C_{e}}^\mp \ket{H} 
\nonumber
\\
=& \sqrt{X_N}^\pm \sum_x (-1)^{m_x} (\mp i)^{m^N_x}\Big[\ket{x}\otimes(\ket{0} + (-1)^{m_x^N} \ket{1})\Big] 
\nonumber
\\
=& \sum_x (-1)^{m_x} (\mp i)^{m^N_x} (\pm i)^{m^N_x \tiny{\mbox{ mod 2}}}\Big[\ket{x}\otimes(\ket{0} + (-1)^{m_x^N} \ket{1}) \Big]
\nonumber
\\
=& \sum_x (-1)^{m_x} (\mp i)^{m^N_x-(m^N_x \tiny{\mbox{ mod 2}})}\Big[\ket{x}\otimes(\ket{0} + (-1)^{m_x^N} \ket{1}) \Big]
\nonumber
\\
=& \sum_x (-1)^{m_x} (-1)^{\binom{m_x^N}{2}}\Big[\ket{x}\otimes(\ket{0}+(-1)^{m_x^N}\ket{1})\Big]. 
\label{proofLC:1}
\end{align}
Eq.~$(\ref{proofLC:1})$ shows that the sign flip of 
$\ket{x}\otimes(\ket{0}+(-1)^{m_x^N} \ket{1}$ is defined 
by $\binom{m_x^N}{2}$. This is nothing but the number of 
pairs of edges in $\mathcal{A}(N)$ that act on $x$. This
sign flip is equivalently described, if we apply the $C_e$  
for all the edges $e \in P$ to the hypergraph. As $C_e^2=\eins$, 
this means that edges in multiset $P$ get complemented. 
\end{proof}

Some examples for the application of this rule are given in Fig.~\ref{fig1}.
Note that the map is not always local, since it contains $\sqrt{C_{e}}^\pm$ 
gates that are nonlocal whenever the vertex $a$ is contained in at least one 
edge of cardinality three or more. Thus, this map can change the entanglement 
properties of the state it is applied to. However, in particular structures 
of hypergraphs, the map can be chosen to be applied to multiple vertices in 
a way that the nonlocal gates cancel each other out. Whenever the nonlocal 
gates cancel each other out we can perform the complementation operation 
without applying those canceling gates at all. Thus the resulting hypergraph 
will be obtained by using local operators only. 

Fig.~\ref{fig1}~(b) and (c) 
display two examples where a sequence of local complementations can be implemented 
using only local operators. These are the first examples that demonstrate that 
two hypergraphs, with edges containing more than two qubits, can be equivalent 
under local unitary operators but not under local Pauli operators. Finally, it
should be noted that our rule of local complementation can also be derived from
the general theory given recently in Ref.~\cite{nika}, but our proof is significantly
simpler.

\section{Permutation unitaries and their applications} 

In the previous section we considered the extension of local 
complementation for hypergraph states. In this section we 
investigate a different family of unitary transformations, 
we call them permutation unitaries. These transformations 
permute the vectors of the computational basis. Such permutations 
are obviously unitary and from Eq.~(\ref{eq-hg-equally})
it is clear that they map hypergraph states to hypergraph states, 
so there must be a graphical description.

The simplest example of such a permutation unitary is Pauli-X (or NOT) 
gate, whose action  on a hypergraph state was studied before 
\cite{Qu2013_encoding, Otfried}, see also Fig.~\ref{fig2} for an example.
A nonlocal example of a permutation unitary in two dimensions 
is a CNOT gate, CNOT$_{ab}:\ket{10}\leftrightarrow\ket{11}$. 
An extension to three-qubit is the Toffoli gate,  CCNOT$_{abc}:\ket{110}\leftrightarrow\ket{111}$. Clearly, is is not necessary
to consider all permutations, as for instance any permutation can
be viewed as a sequence of transpositions \cite{leibniz}.
Another possibility is to look for extensions of CNOT gates.  For two-qubit 
permutations considering  one can easily see that NOT and CNOT are enough 
to cover all possible permutations. Additionally, it is known that every 
permutation on $\{0, 1\}^N$ can be realized by means of a reversible 
circuit using the NOT, CNOT and CCNOT basis and at most one ancilla bit 
\cite{Toffoli}. It is possible to derive a graphical rule of how such maps 
transform hypergraph states. Here we give rules explicitly only for the
two-qubit CNOT and its multiqubit extensions, but the methodology can be 
applied to derive any arbitrary permutation unitary if the exact graphical 
transformation is needed.

\begin{figure}
\begin{center}
\includegraphics[width=0.85\columnwidth]{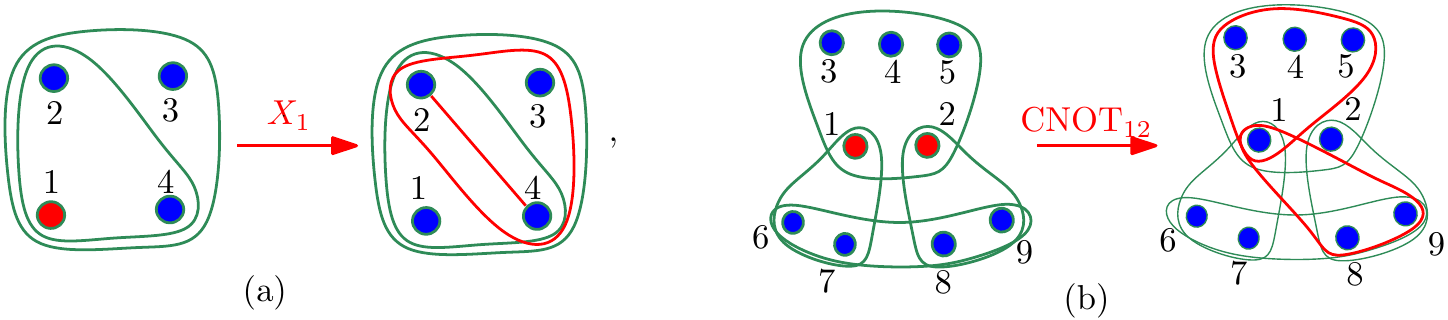}
 \end{center}
\caption{(a) An example of application of Pauli-X or NOT gate on 
the first qubit. We have $\mathcal{A}=\{\{2,4\},\{2,3,4\}\}$.   
(b) An example of application of CNOT$_{12}$, where the first
qubit is the control qubit and the second qubit is the target. 
For the graphical action, we have to consider 
$\mathcal{A}(2)=\{\{1,3,4,5\},\{8,9\}\}$ and in each of its 
elements we have to add the qubit 1. So, the edges
$E_t=\{\{1,3,4,5\},\{1,8,9\}\}$ are added (or removed, if they 
were already present).
}
\label{fig2}
\end{figure}

\begin{lemma}
Applying the  CNOT$_{ct}$ gate on hypergraph state, where $c$ is the 
control qubit and $t$ is the target one, introduces/deletes the  
edges of the form $E_t=\{e_t\cup \{c\}|e_t\in \mathcal{A}(t) \}$.
\end{lemma}
\begin{proof}
Without loss of generality we assume that CNOT$_{12}$ acts on the first 
two qubits. We write a hypergraph state as follows:
\begin{align}
\ket{H}=&\ket{00}\ket{H(E_{00})} \quad 
\quad &  E_{00}=\{e |e\in E, e\cap c =\emptyset, e\cap t =\emptyset\},
\label{CNOT1}
\\
+&\ket{01}\ket{H(E_{00}+E_{01})} \quad 
\quad &  E_{01}=\{e |e\in \mathcal{A}(t), e\cap c =\emptyset\},
\\
+&\ket{10}\ket{H(E_{00}+E_{10})} \quad
\quad &  E_{10}=\{e |e\in \mathcal{A}(c), e\cap t =\emptyset\},
\label{CNOT3}
\\
+&\ket{11}\ket{H(E_{00}+E_{01}+E_{10}+E_{11})} \quad 
\quad &  E_{11}=\{e |e\in \mathcal{A}(\{ c,t\})\}.
\label{CNOT4}
\end{align}
The CNOT$_{12}$ gate swaps $\ket{10}$ and $\ket{11}$, or alternatively 
Eq.~(\ref{CNOT3}) and Eq.~(\ref{CNOT4}), but leaves the other parts 
invariant. Therefore we obtain the following:
\begin{align}
&E_{00}^{\rm new}=E_{00}. 
\label{CNOT_1}  
\\
E_{00}^{\rm new}+E_{01}^{\rm new}=E_{00}+E_{01}    
\quad \Rightarrow \quad 
&E_{01}^{\rm new}=E_{01}.
\label{CNOT_2}
\\
E_{00}^{\rm new}+E_{10}^{\rm new}=E_{00}+E_{01}+E_{10}+E_{11}   
\quad \Rightarrow \quad 
&E_{10}^{\rm new}=E_{01}+E_{10}+E_{11}.
\label{CNOT_3}
\\
E_{00}^{\rm new}+E_{01}^{\rm new}+E_{10}^{\rm new}+E_{11}^{\rm new}=E_{00}+E_{10}    
\quad \Rightarrow \quad 
&E_{11}^{\rm new}=E_{11}.
\label{CNOT_4}
\end{align} 
Equations~(\ref{CNOT_1}-\ref{CNOT_4})  show that  only the edges containing 
the control qubit can appear or disappear. More precisely, Eq.~(\ref{CNOT_3})  
shows that the new edges  that are added/deleted are of the form 
$E_t=\{e_t\cup c | e_t\in\mathcal{A}(t)\}$. 
\end{proof}

An example of this rule is shown in Fig.~\ref{fig2}. We can 
directly generalize this rule to extended CNOT gates, such as the 
Toffoli gate, the proof is essentially the same. 

\begin{corollary}
\label{extendedCNOT}
Applying the extended CNOT$_{Ct}$ gate on a hypergraph state, 
where a set of control qubits $C$ controls the target qubit $t$,  
introduces or deletes the set of edges 
$E_t=\{e_t\cup C|e_t\in\mathcal{A}(t)\}$.
\end{corollary}

Moreover, as mentioned above  every permutation can be constructed
using NOT, CNOT, and CCNOT and at most one ancilla qubit. An ancilla 
qubit is necessary to construct the multiqubit gate set, $\mathcal{T}
= \{{\rm C^{0}NOT}, {\rm CNOT}, \dots,  {\rm C^{k}NOT}\}$ \cite{Xu} 
and the set $\mathcal{T}$ is enough to realize any permutation on
$k$ indices. As $\mathcal{T}$  exactly consists of the gates with 
graphical rules from above, we can state:

\begin{corollary}
Every permutation unitary maps a hypergraph state to a hypergraph state 
and its graphical action can be seen as a composition of rules from 
$\mathcal{T}=\{ {\rm C^{0}NOT}, {\rm CNOT}, \dots ,{\rm  C^{k}NOT} \}$ 
graphical rules.
\end{corollary}

It is interesting to note how the different rules change the
cardinality of edges. If $c$ is the cardinality
of the largest edge in the hypergraph, the NOT gate 
can only create/erase edges with a cardinality strictly 
smaller then $c$.  The CNOT gate can create/erase edges 
with cardinality smaller or equal to $c$, but the CCNOT 
can create edges with cardinality higher then $c$. 

\begin{figure}[t]
\begin{center}
\includegraphics[width=0.95\textwidth]{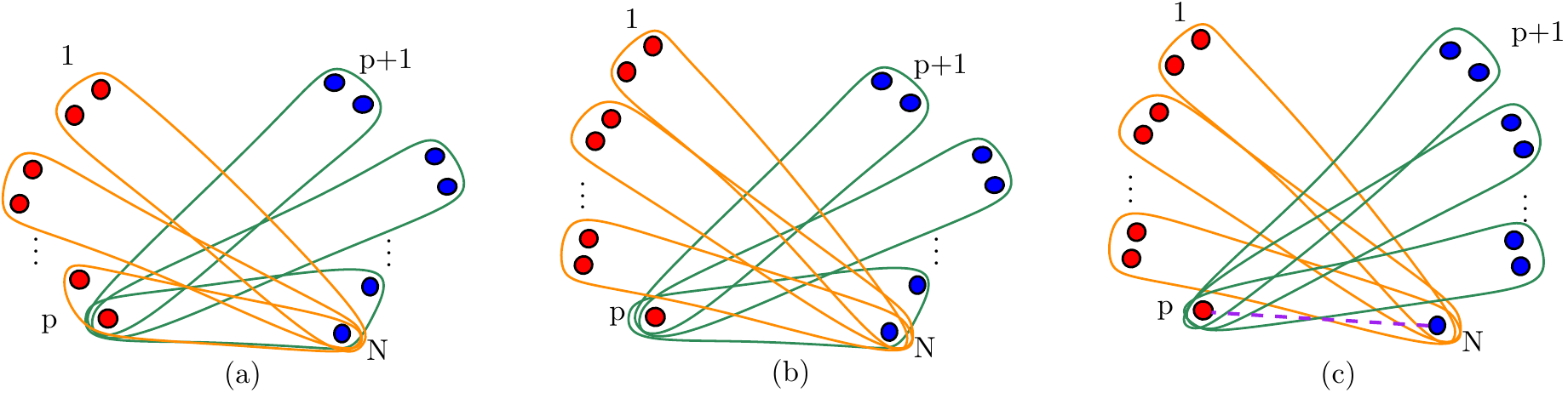}
\end{center}
\caption{Different possibilities of the normal form for complete 
three-uniform hypergraph states. 
(a)  The normal form if $p$ is even and $N$ is even.
(b) The normal form if $p$ is odd and $N$ is odd. 
(c) The normal form if $p$ is odd and $N$ is even. 
We have two cases: 
{(c1)} The hypergraph without the edge $\{p,N\}$. This is the normal 
form if either both  $(p+1)=2 \mod 4$ and $(N-p+1)= 2 \mod 4$ 
or if both $(p+1)=0 \mod 4$ and $(N-p+1)= 0 \mod 4$. 
{(c2)} The hypergraph with the edge $\{p,N\}$.  This is the normal 
form if $(p+1)= 0\mod 4$ and $(N-p+1)=2 \mod 4$ 
or if  $(p+1)=2\mod 4$ and $(N-p+1)=0\mod 4$. 
The edge $\{p,N\}$ is represented by a dashed line.}
\label{fig3}
\end{figure}

Finally, we demonstrate that the rule for the CNOT gate has  
a direct application: Consider a complete three-uniform hypergraph
states, that is, the hypergraph contains all possible three-edges, 
but nothing else. These states can be thought as generalizations 
of GHZ states and violate Bell inequalities in a robust manner 
\cite{gbg}. If one considers a possible bipartition $1,\dots, p|p+1,\dots, N$
of the particles, one may ask how the hypergraph can be simplified using 
unitaries that are local with respect to this bipartition. The following Lemma
provides an answer, and we will use this later for the construction of witnesses.

\begin{lemma}
\label{CNOTREDUCTION}
Consider an $N$-qubit complete three-uniform hypergraph state and
a bipartition $1,\dots, p|p+1,\dots, N$. Then, using only local actions 
with respect to this bipartition the hypergraph can be reduced to 
the form shown in Fig.~\ref{fig3}.  We call this form a normal form 
of complete three-uniform hypergraph state respecting the bipartition 
$1,\dots, p|p+1,\dots, N$.
\end{lemma}

\begin{proof}
The proof consist of an application of a sequence of CNOT gates on both sides
of the bipartition. Details are given in the Appendix. 
\end{proof}

\section{Construction of witnesses}
In this section we consider the construction of entanglement witnesses 
as an application of the results derived so far. More specifically, we construct 
tight witnesses for fully-connected three-uniform hypergraph states. These states 
are of special interest, as it has been shown that they violate Bell inequalities 
with an exponentially increasing amount 
and the violation is robust against particle loss. The Bell inequalities can be used to prove 
that there is some entanglement in the state, but in this section
we will focus on entanglement witnesses for genuine multiparticle entanglement. 

An entanglement witness is an observable which has a non-negative 
expectation value for all separable states, thus, a negative expectation
value signals the presence of entanglement. There are many ways to construct 
entanglement witnesses, see Ref.~\cite{witness} for an overview. One possible
way  to design a witness for a general state $\ket{\psi}$ is to consider the 
following observable
\begin{equation}
\mathcal{W}= \alpha \mathbbm{1} - \ketbra{\psi},
\end{equation}
where $\alpha$ is the maximal overlap between the state $\ket{\psi}$ and 
the pure  biseparable states. This can be computed by the maximal squared 
Schmidt coefficient occurring when computing the Schmidt decomposition with
respect to all bipartitions,
\begin{equation}
\alpha= \max_{\rm bipartitions} 
\big\{
\max_{\lambda_k^{\rm BP}}  \{[\lambda_k^{\rm BP}]^2\}\big\}.
\end{equation} 
For usual graph states the witness can be determined in the following way
\cite{grwitness}: First, for any bipartition one can generate a Bell pair
between the two parties by making only local operations with respect to
this partition. During local operations, however, the maximal Schmidt coefficient
can only increase. This proves directly that for any bipartition 
$\lambda_k^{\rm BP} \leq 1/2$, so $\mathcal{W}= \mathbbm{1}/2 - \ketbra{G}$
is a witness. Using a similar construction, we can estimate $\alpha$
and write down a witness for three-uniform hypergraph states. Note 
that this scheme of constructing witnesses has recently been extended to 
other hypergraph states \cite{newrossi}.
 
\begin{theorem} 
For any three-uniform hypergraph state $\ket{H^3_N}$ the 
operator
\begin{equation}
\label{genwitness}
\mathcal{W}= \frac{3}{4} \mathbbm{1} - \ketbra{H^3_N}
\end{equation}
is an entanglement witness detecting this state. 
\end{theorem}

\begin{proof}
The proof is similar to the one for graph states,  but in this case the aim 
is to share a the three-qubit hypergraph state between the bipartition.  For 
this state the maximal squared Schmidt coefficient is $\alpha={3}/{4}$. 
Given  an $N$-qubit three-uniform state  we consider a bipartition 
$1,\dots, p | p+1,\dots, N$. One can get rid of any edge which entirely 
belongs to either side of the bipartition. Since the graph is assumed to
be connected at least one three-edge remains shared between the two parts. 
Without loss of generality we can assume that this edge is $e=\{p-1,p,p+1\}$. 
Now by making measurements in the Pauli-Z basis on every qubit except these 
three in $e$ we can disentangle all the qubits from the main hypergraph except  
$\{p-1,p,p+1\}$. For all possible measurement results, i.e. with probability 
one the resulting state is, up to local unitarians, a three-qubit hypergraph 
state  consisting only of the edge $e$.
\end{proof}

The previous witness can be used for any connected three-uniform hypergraph
state, but is it no necessarily tight. For the special case of complete
three-uniform states, where any possible three-edge is present, we derive
a better witness in the following. Since this state is symmetric, the Schmidt
coefficients depend only on the size of the partitions.

\begin{lemma}
\label{schmidtlemma}
Consider an $N$-qubit complete three-uniform hypergraph state. Then, 
the maximal squared Schmidt coefficient with respect to the bipartition
$1$ vs.~$N-1$ qubits, is 
\begin{align}
\lambda_1 &=\frac{1}{2} \mbox{ if } N=4k,
\nonumber
\\
\lambda_1&=\frac{1}{2}+\frac{1}{2^{(N+1)/2}} \mbox { if } N=4k+1 \mbox{ or } N=4k+3, 
\nonumber
\\
\lambda_1&=\frac{1}{2}+\frac{1}{2^{N/2}} \mbox{ if } N=4k+2. 
\end{align}
For the $2$ vs.~$N-2$ partition it is  
$\lambda_{2}=\frac{1}{8} (3 + \frac{\sqrt{2^{N+6}+4^N}}{2^N})$. 
For the $3$ vs.~$N-3$ partition it is given by 
$\lambda_3=\frac{9}{16}$ if $N=6$ and for 
$N>6$ one has $\lambda_3<\frac{1}{2}$.
\end{lemma}	
\begin{proof}
The proof is done by tracing out the parties and calculating the Schmidt 
coefficients as eigenvalues of the reduced states. Details can be found
in the Appendix.
\end{proof}

\begin{theorem}
An improved witness for the $N$-qubit complete three-uniform hypergraph 
state $\ket{H_N^{\rm c3}}$ is given by
\begin{equation}
\mathcal{W}=\alpha \mathbbm{1} -\ketbra{H_N^{\rm c3}},
\end{equation}
where $\alpha=\max\{\lambda_1, \lambda_2\}$.
\end{theorem}
\begin{proof}
We have to show that in general it is sufficient to consider the $1$ vs.~$N-1$
and $2$ vs.~$N-2$ partitions. First, the $3$ vs.~$N-3$ give only smaller Schmidt
coefficients, as can be seen from the Lemma~\ref{schmidtlemma}.
For any other $1,\dots, p|(p+1),\dots,N$ bipartition with $p>3$  we 
use  the normal form in Fig.~\ref{fig3}. If a resulting hypergraph is reduced either 
to Fig.~\ref{fig3} (b) or (c) [without the dashed edge], then on qubits $1\dots (p-3)$ 
the measurements in the Pauli-Z basis can be made. As a result, the hypergraph state 
$(p-2),(p-1),p|(p+1)\dots N$ 
is obtained. We know from the Lemma~\ref{schmidtlemma} that the $3$ vs.~$N-3$ 
partition has a largest squared Schmidt coefficient less than $1/2$ (unless $N=6$).  
Keeping in mind that measurements can never decrease the squared maximal Schmidt coefficient, 
we reach the conclusion that  the bipartition $1,\dots, p|(p+1),\dots ,N$  cannot contribute to
the maximal Schmidt coefficient when $p \geq 3$. If in the normal form in Fig.~\ref{fig3} (c)
the dashed edge is present, one can make measurements on both sides of the partition to reduce 
the state to a Bell state between qubits $p$ and $N$. This clearly gives a squared Schmidt
coefficient $\lambda \leq 1/2.$

\begin{figure}
\includegraphics[width=0.77\textwidth, center]{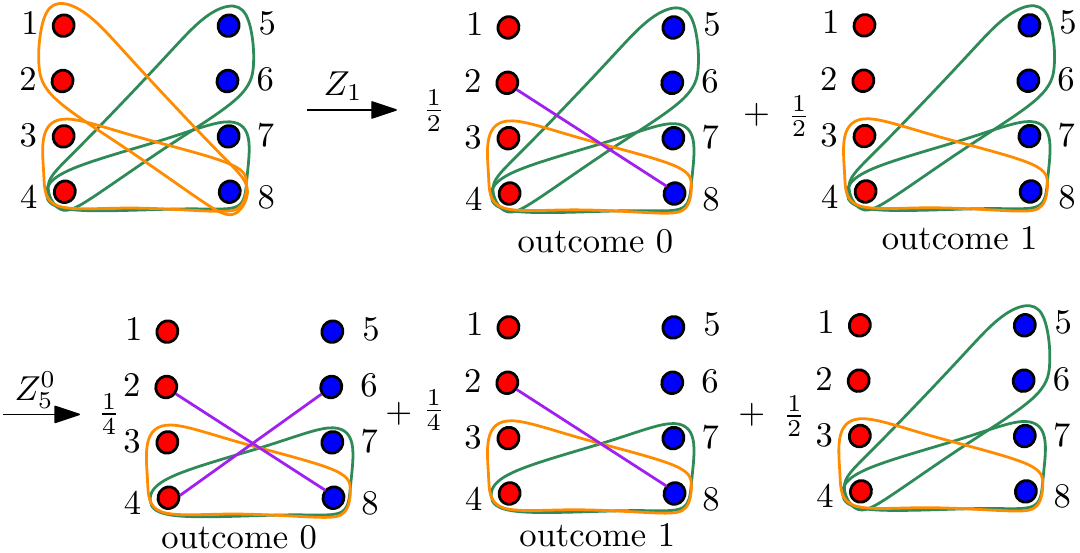}
\caption {Estimation of the Schmidt coefficient for a $4$ vs.~$4$ bipartition. 
See the text for further details.
}
\label{fig4}
\end{figure}

The final case is the state with a normal form in Fig.~\ref{fig3} (a).  Here the 
strategy is as follows: Pauli-Z measurements are made on every qubit but eight of
them, namely the qubits $p-3, p-2, \dots ,p+3,p+4$ remain untouched. This leaves 
us with the state given in Fig.~\ref{fig4}, where the qubits have been relabeled.
Then, a Pauli Z measurement is made on qubit 1. With probability $1/2$  (in case 
of outcome 0),  the edge $\{2,8\}$ is introduced and qubit $1$ is disentangled. 
With probability $1/2$ (outcome $1$) both qubits $1$ and $2$ are disentangled. 
For the first case ($0$ outcome), we again make a Pauli-Z measurement on qubit 
$5$, denoted by $Z_5^0$.  This itself gives two possible outcomes  with half-half 
probabilities, the outcome $0$  gives the edge $\{4,6\}$ and disentangles qubit $5$ 
and  the outcome $0$ disentangles qubits $4$ and $5$. Putting all measurement outcomes 
together with corresponding probabilities  yields as a bound on the Schmidt coefficient
\begin{equation}
\label{boundary}
\lambda \leq 
\frac{1}{4}\cdot\frac{1}{4}
+\frac{1}{4}\cdot\frac{1}{2}\cdot \frac{1}{8}(3+\sqrt{5})
+ \frac{1}{2}\cdot\frac{1}{8}(3+\sqrt{2}) \approx 0.420202 < \frac{1}{2}.
\end{equation}
Note that in this estimation it was used that one minus the largest squared Schmidt 
coefficient can be viewed as the geometric measure of entanglement for this partition, 
and this measure decreases under local operations even for mixed states. 
\end{proof}

\section{Conclusions}

In summary, we have extended the local complementation rule from graph 
states to hypergraph states. We also described the action of different 
gates on hypergraph states with graphical rules. Already for five qubits 
we showed with a simple example that local Pauli operations only are not 
enough to exhaust all local unitary equivalence classes of hypergraph states. 
Based on the rule for the CNOT gate, we developed a normal form for bipartitions 
of complete three-uniform hypergraph states. Based on this, we  derived 
entanglement witnesses for these states.  

There are several directions in which our work can be extended. First, it 
would be highly desirable to develop a general theory for entanglement 
witnesses for hypergraph states, similar to the existing theory for graph 
states \cite{witness}. Here, notions of the coulourability of a hypergraph 
may be developed to characterize how many measurements are needed to 
estimate the fidelity of a state. All this can help to observe hypergraph 
states experimentally. Another interesting question concerns the extent 
to which hypergraph states and their correlations can be simulated 
classically in an efficient manner. Our findings show that certain unitary 
operations have a graphical interpretation. This may be useful to decide 
whether their classical simulation is feasible. For graph states, the 
Gottesman-Knill theorem characterizes a set of operations that can be 
simulated efficiently and it would be highly desirable to identify 
similar operations for hypergraph states.

We thank Cornelia Spee for discussions. This work has been supported 
by the DFG and the ERC (Consolidator Grant 683107/TempoQ).
Additionally, MG would like to acknowledge funding from the Gesellschaft 
der Freunde und F\"orderer der Universit\"at Siegen.

\section{Appendix}

\subsection{Reduction of three uniform hypergraph states to the normal form 
in Lemma 5}

We prove Lemma 5 by considering first simple bipartitions, where the strategy
of the proof is easier to explain. The proof for the general case then follows
the same lines. 

\begin{lemma}
\label{1|N-1}
Consider the bipartition $1|2,3,\dots,N$ for an $N$-qubit complete three-uniform 
hypergraph state.  Then this state is locally (for the given bipartition) equivalent  
to the three-uniform hypergraph state where every vertex is contained in only one 
edge and edges are of the form: 
$E=\{\{1,i,i+1\}\ |\ 2\leq i< N \mbox{ and $i$ is even.}  \}$ 
And only if $N=4k$, an additional cardinality two edge appears, which is $\{1,N\}$. 
\end{lemma}
\begin{proof}
Fig.~\ref{fig5} (a) represents the goal hypergraph state respecting a bipartition 
$1|2,3\dots N$. The algorithm to achieve this state is as follows:
\begin{itemize}

\item[(a)] Erase all the edges which only contain subsets of vertices 
$\{2, 3,\dots N\}$. This operation is local with respect to the bipartition.
\vspace{0.1cm}\\
\textit{ All the remaining edges are $\{\{1,i,j\}| i<j, 2\leq i,j \leq N\}$.}

\item[(b)] Apply $CNOT_{i,i+1}$, where ${2\leq i<N}$. 
\\
\textit{To give an example, we start with the  $CNOT_{23}$  gate.  The adjacency 
of $3$ is  $\mathcal{A}(3)=\{1, i\}$, where $i\neq 3 , 2\leq i \leq N.$   The 
edges introduced by the $CNOT_{23}$  gate are $\{e_t \cup \{2\}| e_t\in\mathcal{A}(3)\}$ 
and therefore, this action removes all the edges where $2$ is contained except the 
edge $\{1,2,3\}$ and adds the cardinality two edge $\{1,2\}$.  
\vspace{0.1cm} \\ 
At this step the remaining edges are $\{\{1,2\},\{1,2,3\},\{1,i,j\}| i<j, 3\leq i,j \leq N\}$.
Since  $2\notin \mathcal{A}(i+1), i\geq 3$, it is clear that consecutive  $CNOT$ gates 
presented in this step  do not modify edges containing $2$. $CNOT_{34}$ erases all the 
edges where $3$ is presented except already established  $\{1,2,3\}$ and  the gate where 
vertices $3$ and $4$ are presented together $\{1,3,4\}$. It also adds the cardinality 
two edge $\{1,3\}$. Repeating this procedure: 
\vspace{0.1cm}\\
All the remaining edges are of the form $\{\{1,j,j+1\} | 2\leq j < N \} $ 
and cardinality two edges $\{\{1, i\}|2\leq i< N\}$.}

\item[(c)] Apply $CNOT_{i+2,i}$, where $2\leq i< N -1\mbox{ and $i$ is even.} $
\\
\textit{The adjacency of $i=2$ mod $4$ right before applying the  $CNOT_{i+2,i}$ 
gate  is $\mathcal{A}(i)=\{\{1\}, \{1, i+1\}\}$.   The  $CNOT_{i+2,i}$ gate, 
therefore, erases/creates edges $\{\{1,i+2\}, \{1, i+1, i+2\}\}$. This means 
that  the adjacency of $i=0$ mod $4$ is only $\{1, i+1\}$, and the  
$CNOT_{i+2,i}$ gate can only erase $\{1, i+1, i+2\}$. See Fig.~\ref{fig5}~(b).
\\
Here we have to consider several cases: 
\vspace{0.2cm}\\
$(1)$  $N$ is odd: All the remaining edges are of the form $\{1, i, i+1\}$, 
for  even $2\leq i \leq N$ and also $\{1, i\}$ for unless $i=0$ mod $4$. It 
is easy to see that all two edges can be removed by  action of Pauli-X's.   
\vspace{0.2cm}\\
$(2)$ $N$ is even: If $N=2$ mod $4$, then the the last edge $\{1, N-1,N \}$ 
is erased and the edge $\{1,N\}$ cannot be created. Therefore, the last qubit 
is completely disentangled in this case. See Fig.\ref{fig5} (a).
\vspace{0.2cm}\\
$(3)$ $N$ is even: If $N=0$ mod $4$, then the the last edge $\{1, N-1,N \}$ 
is erased and the edge $\{1,N\}$ is created. See Fig.\ref{fig5} (b) for 
the exact procedure.}
\end{itemize}
\end{proof}

To sum up the previous theorem, there are three possibilities for the final 
hypergraph and it only depends on the number of parties in the hypergraph. 
If $N$ odd, then every vertex is exactly in one hyperedge. If $N=4k$, 
then the final hypergraph corresponds to the one in Fig.~\ref{fig5}~(a) 
including the dashed line.  Note that this is in line with the fact that
the maximal Schmidt coefficient for this case is $1/2$  (see Lemma 7),
as there is a Bell pair shared across the bipartition \cite{grwitness}.  
In  case $N=4k+2$, the dashed line is missing, therefore, the last qubit 
can be removed and the result for the maximal Schmidt coefficient 
matches with $N=4k+1$ case. 

\begin{figure}
\begin{center}
\includegraphics[width=0.75\columnwidth]{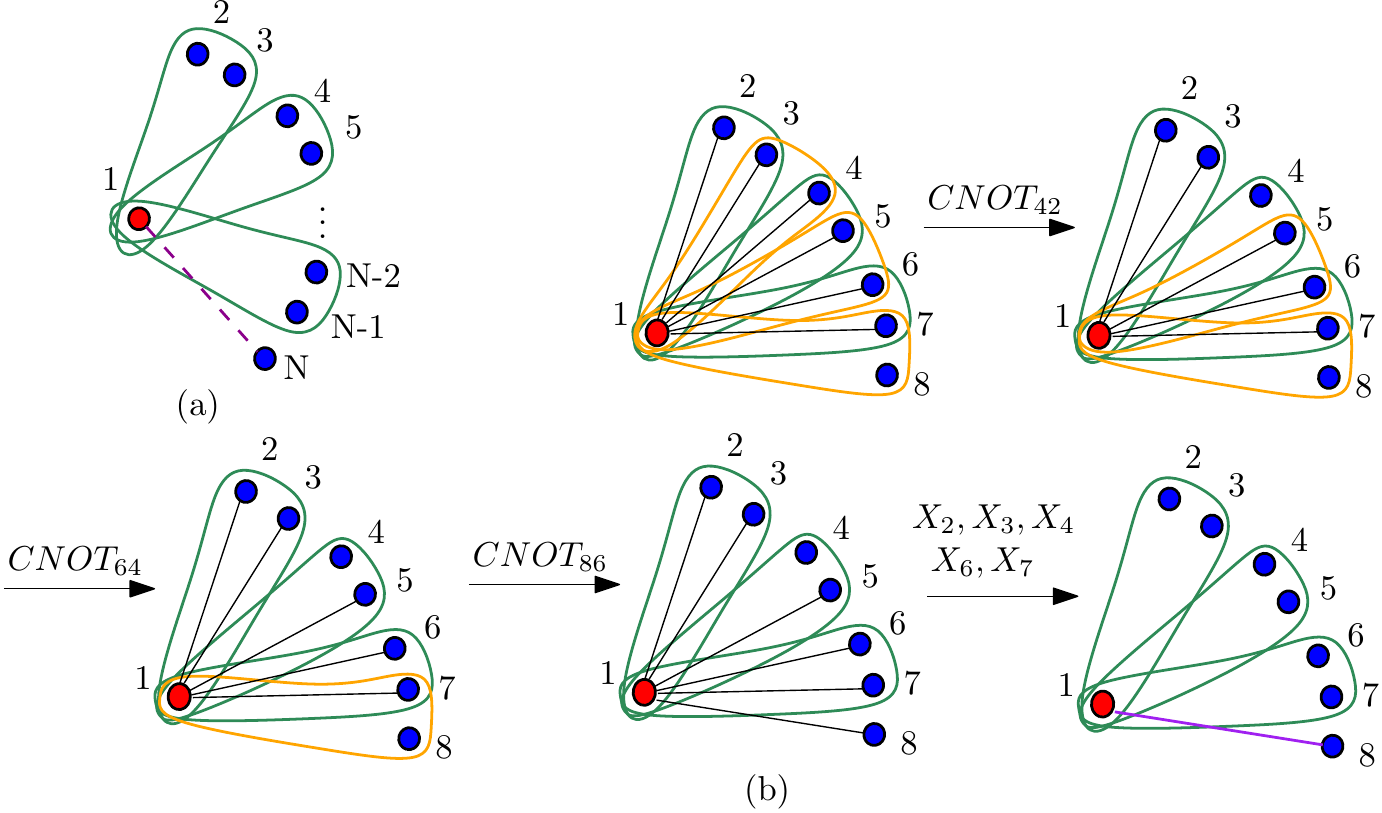}
\end{center}
\caption{(a) The target graph for Lemma 9. (b) An example 
for step (c) of the proof of Lemma 9. See the text for 
further details.}
\label{fig5}
\end{figure}

\begin{lemma}
Considering the bipartition $12|3,4,\dots,N$ for an $N$-qubit fully-connected three-uniform 
hypergraph state. Then, this state is locally  equivalent  to the three-uniform hypergraph 
state already derived from the bipartition $2|3\dots N$ and in addition has the edge 
$\{1,2,N\}$.
\end{lemma}
\begin{proof}
The steps are very similar to the $1|23\dots N$  case:
\begin{itemize}
\item[(a)] {From} a  hypergraph we remove all the edges which do not contain either 
party $1$, or $2$.
\\
\textit{All the remaining edges are $\{\{1, i, j\},\{2, i, j\},\{1,2,i\}| 2<i<j\leq N\}$}
\item[(b)]  Apply $CNOT_{12}$. 
\\
\textit{The adjacency of $2$ is $\mathcal{A}(2)=\{\{i,j\}, \{1, i\}\}$ for 
$ 2<i<j\leq N\}$. Augmented by the control qubit $1$, action of $CNOT_{12}$ 
removes $\{1, i, j\}$ and creates $\{1, i\}$, $3\leq i\leq N$.} 
\\
\textit{All the remaining edges are 
$\{\{1, i\},\{2, i, j\},\{1,2,i\}| 2<i<j\leq N\}$}
\item[(c)] Apply $X_2$. 
\\
\textit{All the elements  in the adjacency $\mathcal{A}(2)=\{\{1,i\},\{i,j\}\}$ 
for $ 2<i<j\leq N\}$  are added as edges to the hypergraph. Thus, all the edges 
of the type $\{1,i\}$ cancel out and $\{i,j\}$ can be directly removed.}
\\
\textit{All the remaining edges are $\{\{2, i, j\},\{1,2,i\}| 2<i<j\leq N\}$}
\item[(d)] Apply the algorithm from Lemma \ref{1|N-1} to the qubits $2|3\dots N$.
\\
Each of the action of the $CNOT_{i, i+1}$ gate  from  Lemma \ref{1|N-1}  step 
(b),  $3\leq i <N$, removes the edge $\{1,2,i\}$ in addition to the its actions 
considered in Lemma \ref{1|N-1}. Therefore only the edge  $\{1,2,N\}$ remains 
of this type and other ones resulting from the action algorithm from Lemma 
\ref{1|N-1} on qubits $2|3\dots N$.
\end{itemize} 
\end{proof}

\begin{corollary}
Considering any bipartition $1\dots p|(p+1)\dots N$  of the complete three-uniform 
hypergraph state, the consecutively applying local CNOT gates (respecting the 
bipartition)  reduces the hypergraph to the union of two hypergraphs as shown on  
Fig.~\ref{fig3}: $p|p+1\dots N$ and $N|1\dots p$, both already reduced to the normal 
form by the Lemma \ref{1|N-1}. 
\end{corollary}

This result is obtained by applying the algorithm from Lemma \ref{1|N-1} first to 
$N|1\dots p$ and then to $p|(p+1)\dots N $. This ends the proof of Lemma 5.

\subsection{Proof of Lemma~\ref{schmidtlemma}:}

\begin{proof}
First we consider $1$~vs.~$N-1$ bipartition. To calculate the maximal 
Schmidt coefficient we compute the reduced density matrix. As the state 
is symmetric, we only have to take the bipartition $1|2,3,\dots,N$. We 
have 
\begin{equation}
\varrho_1 = 
Tr\big( \ketbra{H} \big)_{2\dots N}
=\frac{1}{2^N}
\begin{pmatrix}
2^{N-1}& a \\
a & 2^{N-1} 
\end{pmatrix}.
\end{equation}
The diagonal elements follow directly from the representation of the 
hypergraph state in Eq.~(\ref{eq-hg-equally}) and do not depend on the
structure of the hypergraph. 
For computing the off-diagonal entries, we write the hypergraph state 
as
\begin{equation}
\ket{H}=\ket{0}\sum_x\Big[(-1)^{f_0(x)}\ket{x} \Big]+
\ket{1}\sum_x \Big[(-1)^{f_1(x)}\ket{x} \Big].
\end{equation}
with $x \in \{0,1\}^{(N-1)}.$ Since we deal with three-uniform complete 
hypergraph states, we have $f_0 = {\binom{w(x)}{3}}$ and  
$f_1={\binom{w(x)+1}{3}}$, where $w(x)$ is the weight (i.e., the number
of ``1'' entries) of $x$. We can then write
\begin{align}
\label{eq-a}
a &=\sum_x (-1)^{f_0(x)+f_1(x)}.
\end{align}
The values of $f_0$ and $f_1$ do only depend on $w(x) \mod 4$.
Instead of summing over $x$, we can also sum over all  possible $k=w(x)$ 
in Eq.~(\ref{eq-a}) and distinguish the cases of  $k\mod 4.$ The value
for a given $k$ is then up to the sign given by the numbers of possible
$x$ with the same $w(x)=k.$ We have:
\begin{align}
a=
&\sum_{k=0,4\dots}^{2^{N-1}}
\bigg[
\binom{N-1}{k}+\binom{N-1}{k+1}-\binom{N-1}{k+2}-\binom{N-1}{k+3}
\bigg]
\nonumber
\\
\label{eq-f(x1)}
=& Re \big[(1+i)^{N-1}\big]+Im\big[(1+i)^{N-1}\big].
\end{align}
To give the final result, we have to consider several cases in Eq.~(\ref{eq-f(x1)}): 
If $N=4\ell$, then $a=0$, therefore,  $\varrho_1$  is  maximally mixed and $\lambda_1=1/2.$ 
If $N=4\ell+1$ or  $N=4\ell+3$,  then $a=\pm 2^{\frac{N-1}{2}}$. Then, it follows that 
$\lambda_1=1/2+1/2^{\frac{N+1}{2}}$. Similarly, for $N=4\ell+2$, $a=\pm 2^{\frac{N}{2}}$  
and  therefore   $\lambda_1=1/2+1/2^{\frac{N}{2}}$. This ends the computation of $\lambda_1$.

Second, we look at the $2$~vs.~$N-2$ bipartitions. The idea of the proof very much 
resembles the previous case. First, we take the bipartition $1,2|3,4,\dots, N$ and 
trace out the second part:
\begin{equation}
\varrho_{12}=Tr\big( \ketbra{H} \big)_{3\dots N}=\frac{1}{2^N}
\begin{pmatrix}
2^{N-2}& a_{+} & a_{+} & 0 \\
  a_{+} & 2^{N-2}& 2^{N-2} &a_{-}\\
  a_{+}&2^{N-2}&2^{N-2}&a_{-}\\
 0 &a_{-}&a_{-}&2^{N-2}
 \end{pmatrix}.\end{equation}
For computing the entries, we express a hypergraph state in the following 
way:
\begin{align}
\ket{H} & =\ket{00}\sum_x\Big[(-1)^{f_{00}(x)}\ket{x} \Big]
+\ket{01}\sum_x\Big[(-1)^{f_{01}(x)}\ket{x} \Big]
\nonumber
\\
&+\ket{10}\sum_x\Big[(-1)^{f_{10}(x)}\ket{x} \Big]
+\ket{11}\sum_x\Big[(-1)^{f_{11}(x)}\ket{x} \Big].
\end{align} 
with $x \in \{0,1\}^{(N-2)}.$ The diagonal elements of $\varrho_{12}$ are, as before, 
easy to determine. This is also the case for the two anti-diagonal terms 
$\ket{01}\bra{10}= \ket{10}\bra{01}$ as $f_{01}+f_{10}$ is always even. 
The next term, $a_{+}$, is derived as Eqs.~(\ref{eq-a}, \ref{eq-f(x1)}):  
$a_{+}=Re[(1+i)^{N-2}]+Im[(1+i)^{N-2}]. $ For the term $\ket{00}\bra{11}$, 
$f_{00}(x)+f_{11}(x)$ is even if $w(x)$ is even and is odd if $w(x)$ is odd. 
Therefore  $\sum_x (-1)^{f_{00}+f_{11}}=0 $.  For the last term we find 
$a_{-}=a-a_{+}=Re\big[(1+i)^{N-2}\big]-Im\big[(1+i)^{N-2}\big]$.

Putting all these terms together in the matrix, one can calculate 
the maximal eigenvalue of $\varrho_{12}$:
\begin{align}
\lambda_2 & =\frac{1}{8} (3 + \frac{\sqrt{4^N + 128 (a_{+}^2 +a_{-}^2)}}{2^N} )
=\frac{1}{8} (3 + \frac{\sqrt{4^N + 64 Abs[(1 + i)^N]^2}}{2^N})
\nonumber
\\
&=\frac{1}{8} (3 + \frac{\sqrt{4^N+2^{N+6}}}{2^N}).
\end{align}

Finally, we have to consider the $1,2,3|4\dots N$ bipartition and write down the 
reduced density matrix: 
\begin{equation}
\varrho_{123}=Tr\big( \ket{H}\bra{H} \big)_{4\dots N}=\frac{1}{2^N}
\begin{pmatrix}
2^{N-3}& c & c & 0 & c & 0& 0& b \\
c &2^{N-3} &2^{N-3} & -b & 2^{N-3} & -b& -b& 0 \\
c &2^{N-3} &2^{N-3} & -b & 2^{N-3} & -b& -b& 0 \\
0 &-b &-b &  2^{N-3} & -b &  2^{N-3}&  2^{N-3}& -c \\
c &2^{N-3} &2^{N-3} & -b & 2^{N-3} & -b& -b& 0 \\
0 &-b &-b &  2^{N-3} & -b &  2^{N-3}&  2^{N-3}& -c \\
0 &-b &-b &  2^{N-3} & -b &  2^{N-3}&  2^{N-3}& -c \\
b &0 &0 & -c &0 & -c&  -c&2^{N-3} 
 \end{pmatrix},
 \end{equation}
where $c=Re\big[(1+i)^{N-3}\big]-Im\big[(1+i)^{N-3}\big]$ 
and $b=-Re\big[(1+i)^{N-3}\big]+Im\big[(1+i)^{N-3}\big]$.

{From} this we can be derive all possible values of maximal Schmidt coefficient 
$\lambda_3$. 
If $N \equiv 4k$, then $\lambda_3=2^{-2}+2^{-4k-3}\sqrt{48\cdot 2^{4k}+4^{4k}}$. 
If $N\equiv 4k+1$, then $\lambda_3=2^{-2}+2^{-2k-1}+2^{-4k-3}\sqrt{16^k(16+2^{4k}+2^{2k+2})}$.  
If $N\equiv 4k+2$, then $\lambda_3=2^{-2}+2^{-2k-1}+2^{-4k-3}\sqrt{2^{4k}(2+2^{2k})^2}$ and 
finally, if $N\equiv 4k+3$, then $\lambda_3=2^{-2}+2^{-2k-2}+2^{-4k-3}\sqrt{2^{4k}(4+2^{4k}+2^{2k+1})}$.
It can be easily seen that $\lambda_3$ is decreasing with $N$ and it is only greater than $1/2$ when $N=6$. 
\end{proof}

\end{document}